\documentclass[11pt]{article}
\usepackage{amsmath,amssymb,amsfonts,amsthm,epsfig}
\usepackage[usenames,dvipsnames]{xcolor}

\usepackage{bm,xspace}

\usepackage{cancel}
\usepackage{fullpage}
\usepackage{liyang}
\usepackage{framed}
\usepackage{verbatim}
\usepackage{enumitem}
\usepackage{array}
\usepackage{multirow}
\usepackage{afterpage}
\usepackage{mathrsfs}  

\usepackage{dsfont}
\renewcommand{\F}{\mathds{F}}

\usepackage[normalem]{ulem}

\newcommand{\rnote}[1]{\footnote{{\bf \color{red}Rocco}: {#1}}}

\usepackage{caption} 

\usepackage{todonotes}

\makeatletter
\newtheorem*{rep@theorem}{\rep@title}
\newcommand{\newreptheorem}[2]{
\newenvironment{rep#1}[1]{
 \def\rep@title{#2 \ref{##1}}
 \begin{rep@theorem}\itshape}
 {\end{rep@theorem}}}
\makeatother
\theoremstyle{plain}



\def\colorful{1}

\ifnum\colorful=1

\newcommand{\red}[1]{{\color{red} {#1}}}

\fi
\ifnum\colorful=0

\newcommand{\red}[1]{{{#1}}}

\fi

\usepackage{boxedminipage}

\newcommand{\ignore}[1]{}

\newreptheorem{theorem}{Theorem}
\newtheorem*{theorem*}{Theorem}
\newreptheorem{lemma}{Lemma}
\newreptheorem{proposition}{Proposition}
\newtheorem*{noclaim*}{Claim}

\newcommand{\DT}{\mathsf{DT}}

\newcommand{\uhr}{\upharpoonright}

\newcommand{\fixed}{\mathrm{fixed}}

\newcommand{\acz}{\mathsf{AC^0}}

\newcommand{\comm}{\mathrm{comm}}
\newcommand{\err}{\mathrm{err}}
\newcommand{\corr}{\mathrm{corr}}
\newcommand{\target}{\mathrm{target}}

\newcommand{\SL}{\mathrm{SL}}

\newcommand{\SYM}{\mathsf{SYM}}
\newcommand{\THR}{\mathsf{THR}}
\newcommand{\ANY}{\mathsf{ANY}}
\newcommand{\GIP}{\mathrm{GIP}}
\newcommand{\RW}{\mathrm{RW}}

\newcommand{\G}{\mathsf{G}}

\newcommand{\pparagraph}[1]{\medskip \noindent {\bf {#1}}}

\begin{document}

\title{Luby--Veli{\v{c}}kovi{\'c}--Wigderson revisited: \\
Improved correlation bounds and pseudorandom generators \\ for depth-two circuits}



\author{ Rocco A.~Servedio\thanks{Supported by NSF grants CCF-1420349 and CCF-1563155. {\tt rocco@cs.columbia.edu}}\\ 
Columbia University  \and Li-Yang Tan\thanks{Supported by NSF grant CCF-1563122.  Part of this research was done during a visit to Columbia University. {\tt liyang@cs.columbia.edu}} \\ Toyota Technological Institute }

\begin{titlepage}

\maketitle

\begin{abstract}

We study correlation bounds and pseudorandom generators for depth-two circuits that consist of a $\SYM$-gate (computing an arbitrary symmetric function) or $\THR$-gate (computing an arbitrary linear threshold function) that is fed by $S$ $\AND$ gates.  Such circuits were considered in early influential work on unconditional derandomization of Luby, Veli{\v{c}}kovi{\'c}, and Wigderson \cite{LVW93}, who gave the first non-trivial PRG with seed length $2^{O(\sqrt{\log(S/\eps)})}$ that $\eps$-fools these circuits.

In this work we obtain the first strict improvement of~\cite{LVW93}'s seed length: we construct a PRG that $\eps$-fools size-$S$ $\{\SYM,\THR\} \circ\mathsf{AND}$ circuits over $\zo^n$ with seed length 
\[ 2^{O(\sqrt{\log S })} +  \polylog(1/\eps), \]
an exponential (and near-optimal) improvement of the $\eps$-dependence of \cite{LVW93}. The above PRG is actually a special case of a more general PRG which we establish for constant-depth circuits containing multiple $\SYM$ or $\THR$ gates, including as a special case $\{\SYM,\THR\} \circ \acz$ circuits. These more general results strengthen previous results of Viola \cite{Vio06} and essentially strengthen more recent results of Lovett and Srinivasan \cite{LS11}. 

Our improved PRGs follow from improved correlation bounds, which are transformed into PRGs via the Nisan--Wigderson ``hardness versus randomness'' paradigm \cite{NW94}.  The key to our improved correlation bounds is the use of a recent powerful \emph{multi-switching} lemma due to H{\aa}stad \cite{Has14}.

\end{abstract}

\thispagestyle{empty}

\end{titlepage}

\section{Introduction} 

Depth-2 circuits which have a  $\SYM$ or $\THR$ gate at the output and $\AND$ gates (of arbitrary fan-in) adjacent to the input variables are central objects of interest in concrete complexity, lying at the boundary of our understanding for many benchmark problems such as lower bounds, learning, and pseudorandomness.  The class of $\SYM \circ \AND$ circuits (also known as $\SYM^+$ circuits) has received much attention even in the restricted case of $\polylog(n)$ bottom fan-in because of the well-known connection with the complexity class $\mathsf{ACC}^0$~\cite{Yao:90,BeigelTarui:94,CP16}, a connection that is at the heart of Williams's breakthrough circuit lower bound~\cite{Williams11ccc} showing $\mathsf{NEXP} \ne \mathsf{ACC}^0$.  Another well-studied subclass, corresponding to the special case where the $\SYM$ gate computes the parity of its inputs, is the class of $S$-sparse polynomials over $\F_2$, which have been intensively studied in a wide range of contexts such as learning \cite{SchapireSellie:96,Bshouty:97ipl,BshoutyMansour:02},
approximation and interpolation
\cite{Karpinski:89,GKS:90,rothben:91}, deterministic
approximate counting \cite{EK89,KL93,LVW93}, and property testing
\cite{DLM+:07,DLMW10:algorithmica}.  Turning to $\THR$ gates (which compute an arbitrary linear threshold function of their inputs) as the top gate, the class of $\THR\circ \AND$ circuits of size-$S$ is easily seen to contain the class of $S$-sparse polynomial threshold functions over $\zo^n$.  This class, and special cases of it such as low-degree polynomial threshold functions, has also been intensively studied in complexity theory, learning theory, and derandomization, see e.g.~\cite{MinskyPapert:68,Goldmann:97,KrausePudlak:98,KKMS:08,Podolskii:09,MZstoc10,DKNfocs10,DOSW:11,Kane12,DS14stoc} and many other works. In this work we focus on \emph{pseudorandom generators} for these $\{\SYM, \THR\} \circ \AND$ circuits. 

In 1993 Luby, Veli{\v{c}}kovi{\'c}, and Wigderson~\cite{LVW93} gave the first pseudorandom generators for these depth-2 circuits.  As we shall discuss in detail below, this result was subsequently extended in various ways by different authors, but prior to the present work no strict improvement of Theorem~\ref{thm:LVW} was known for the class of circuits that it addresses.

\begin{theorem}[Luby--Veli{\v{c}}kovi{\'c}--Wigderson 1993]  \label{thm:LVW}
There is a PRG with seed length $2^{O(\sqrt{\log(S/\eps)})}$ that $\eps$-fools the class of size-$S$ $\SYM \circ \AND$ circuits over $\zo^n$. The same is true for the class of size-$S$ $\THR \circ \AND$ circuits.\footnote{\cite{LVW93} does not actually consider $\THR \circ \AND$ circuits, but as we discuss later their arguments also apply to this class. } \end{theorem} 

The main contribution of the present work is an exponential improvment of Theorem~\ref{thm:LVW}'s dependence on $\eps$, giving the first strict improvement of the \cite{LVW93} seed length: 

\begin{theorem}[Our main result] \label{thm:main}
There is a PRG with seed length $2^{O(\sqrt{\log S})} + \polylog(1/\eps)$ that $\eps$-fools the class of size-$S$ $\SYM \circ \acz$ circuits. The same is true for $\THR \circ \acz$ circuits.
\end{theorem}

Theorem~\ref{thm:main} improves on a result of Viola \cite{Vio06} which, building on \cite{LVW93}, gave a $2^{O(\sqrt{\log (S/\eps)})}$-seed-length PRG for size-$S$ $\SYM \circ \acz$ circuits.  The \cite{Vio06} PRG combines correlation bounds against $\SYM \circ \acz$ circuits with the Nisan--Wigderson  ``hardness versus randomness'' paradigm, which yields pseudorandom generators from correlation bounds; we similarly prove Theorem \ref{thm:main} by establishing improved correlation bounds and using the Nisan--Wigderson paradigm.  

\pparagraph{Near-optimal hardness-to-randomness conversion.}
A major theme in computational complexity over the the last several decades, dating back to the seminal works of~\cite{Sha81,Yao82,BM84,Nis91,NW94}, has been that \emph{computational hardness} can be converted into \emph{pseudorandomness}. This insight is at the heart of essentially all unconditional pseudorandom generators, and motivates the goal of understanding when and how this conversion can be carried out in a quantitatively optimal manner.  With this perspective in mind, we observe that the dependence on $\eps$ in Theorem~\ref{thm:main} is optimal up to polynomial factors, and as we discuss in Section~\ref{sec:barriers}, achieving better dependence on $S$ even for the special case of $\{\SYM, \THR\} \circ \AND$ circuits would require groundbreaking new lower bounds against low-degree $\F_2$ polynomials and $\mathsf{ACC}^0$ circuits. Hence Theorem~\ref{thm:main} achieves a near-optimal hardness-to-randomness conversion for $\{\SYM, \THR\} \circ \acz$ circuits; the seed length of our  PRG is essentially the best possible given current state-of-the-art correlation bounds and circuit lower bounds. 

The exponential improvement in $1/\eps$ over~\cite{LVW93}'s seed length translates immediately into significantly improved deterministic approximate counting and deterministic search algorithms for $\{ \SYM, \THR\} \circ \acz$ circuits, two basic algorithmic tasks in unconditional derandomization\footnote{Indeed, the work of~\cite{LVW93} was explicitly motivated by deterministic approximate counting of $S$-sparse $\F_2$ polynomials; see the abstract of~\cite{LVW93}.} (see e.g.~\cite{AjtaiWigderson:85} for formal definitions of these tasks and a discussion of how PRGs yield deterministic algorithms for them).

In the rest of this introduction we provide background and context for our results and explain the main ingredients that underlie them.

\subsection{Prior PRGs and correlation bounds for $\{\SYM,\THR\} \circ \acz$} 

As mentioned above, the first results on PRGs for $\SYM \circ \AND$ circuits were given in early influential work of Luby, Veli{\v{c}}kovi{\'c}, and Wigderson~\cite{LVW93}, who constructed a PRG that $\eps$-fools size-$S$ $\SYM \circ \AND$ circuits over $n$ variables with seed length $2^{O(\sqrt{\log (S/\eps)})}$.  The work of \cite{LVW93} employed ideas similar to those in the ``hardness versus randomness'' paradigm of \cite{NW94}, which subsequently came to be well understood as a versatile technique for constructing pseudorandom generators from correlation bounds.

\ignore{
\cite{LVW93} was done before the~\cite{NW94} hardness-versus-randomness paradigm really crystalized, was not explicit then. Their arguments employed the ideas in this framework (already present in~\cite{Nis91}'s work) but more complicated. }

A number of years later, with the~\cite{NW94} framework in hand, Viola~\cite{Vio06} made the useful observation that correlation bounds against the larger class of $\SYM \circ \AND \circ \OR$ circuits translate to PRGs for $\SYM \circ \AND$ circuits in a ``black-box'' manner via~\cite{NW94}, and the same is true when the top gate is $\THR$ instead of $\SYM$.  
(Informally, the~\cite{NW94} translation ``costs'' two layers of depth:  with typical parameter settings, it yields PRGs for a class $\calC$ from correlation bounds against $\calC \circ \ANY_{\log n}$ circuits, where an $\ANY_t$ gate computes an arbitrary $t$-variable Boolean function.  By rewriting the $\ANY_{\log n}$ gate as a CNF, it is possible to collapse the two adjacent layers of $\AND$ gates, yielding Viola's observation.)  Roughly speaking, in this translation from correlation bounds against $\calC \circ \ANY_{\log n}$ to PRGs that $\eps$-fool $\calC$,

\begin{itemize}

\item the larger the $\calC \circ \ANY_{\log n}$ circuits for which the correlation bound holds,  the better (smaller) is the PRG's seed length for fooling size-$S$ functions in $\calC$; and

\item the smaller the advantage over random guessing that the correlation bound establishes, the better (smaller) is the PRG's seed length's dependence on the fooling parameter~$\eps$.

\end{itemize}

Motivated by this template, \cite{Vio06} established $n^{-\Omega(\log n)}$ correlation bounds against $\SYM\circ\acz$ circuits of size $n^{\Omega(\log n)}$.  This translates (see Appendix \ref{sec:NW}) into a PRG with seed length 
$2^{O(\sqrt{\log(S/\eps)})}$ for size-$S$ $\SYM \circ \acz$ circuits over $\zo^n$, matching the seed length achieved by \cite{LVW93} but for a larger class of circuits (and also with a simpler and more modular proof).  While \cite{Vio06} does not explicitly discuss $\THR$ gates, his proof like that of \cite{LVW93} also goes through for $\THR \circ \acz$ as remarked in the earlier footnote.

Subsequent work of~\cite{LS11} established a strong correlation bound of $\exp(-\Omega(n^{1-o(1)}))$ against $\SYM \circ \acz$ and a correlation bound of $\exp(-\Omega(n^{1/2 - o(1)}))$ against $\THR \circ \acz$, but in both cases only for such circuits of size $n^{O(\log\log n)}$. Via the Nisan--Wigderson framework \cite{NW94} this translates into a PRG with seed length $2^{O(\log S/\log\log  S)} + \polylog(1/\eps)$ for size-$S$ $\{\SYM, \THR\} \circ \acz$ circuits over $\zo^n$; while this is a very good dependence on $\eps$, it comes at the cost of a significantly worse dependence on the circuit size $S$.  Thus both the seed length and correlation bounds of \cite{LS11} are incomparable to those of~\cite{LVW93,Vio06}; see Table~1.

(We further note that other incomparable results have been achieved in separate lines of work  on pseudorandom generators for degree-$d$ polynomial threshold functions~\cite{DKNfocs10,MZstoc10,Kane12} and degree-$d$ $\F_2$ polynomials~\cite{Bog05,BV10,Lov09,Vio09}, which correspond to $\THR \circ \AND_d$ and $\PAR \circ \AND_d$ circuits respectively. The seed lengths of these PRGs all have an exponential dependence on $d$, and thus do not yield non-trivial results for general $\poly(n)$-size $\THR \circ \AND$ or $\PAR \circ \AND$ circuits. For constant $d$, the~\cite{Lov09,Vio09} PRGs for $\PAR \circ \AND_d$ circuits achieve optimal seed length, while the~\cite{MZstoc10,Kane12} PRGs for $\THR\circ \AND_d$ have seed length $\poly(1/\eps)\cdot \log n$.)

\begin{table}[t]
\renewcommand{\arraystretch}{1.6}
\centerline{
\begin{tabular}{|c|c|c|c|c|}
\hline
&  Circuit type &  Circuit size $S$ & Correlation bound & PRG seed length\\ \hline
\cite{Vio06} & $\{\SYM,\THR\} \circ \acz_d$ & $n^{ c_d \log n}$ & $n^{-c_d \log n}$  & $2^{O(\sqrt{\log (S/\eps)})}$ \\ \hline 
\cite{LS11} & $\SYM \circ \acz_d$ & $n^{c_d \log \log n}$ & $\exp(-n^{1-o(1)})$   & $2^{O\big({\frac {\log S}{\log \log S}}\big)} + (\log (1/\eps))^{2+o(1)}$ \\ \hline
\cite{LS11} & $\THR \circ \acz_d$ & $n^{c_d \log \log n}$ & $\exp(-n^{1/2-o(1)})$   & $2^{O\big({\frac {\log S}{\log \log S}}\big)} + (\log (1/\eps))^{4 + o(1)}$ \\ \hline \hline
{\bf This work} & $\{\SYM,\THR\} \circ \acz_d$ & $n^{ c \log n}$ & $\exp(-\Omega(n^{0.499}))$     & $2^{O(\sqrt{\log S})} + (\log(1/\eps))^{4.01}$ 
\\ \hline
\end{tabular}
}
\caption{Correlation bounds against $\{\SYM,\THR\} \circ \acz_d$ circuits and the PRGs that follow via the \cite{NW94} paradigm. In all cases the ``hard function'' is the $\RW$ function that is defined in (\ref{eq:RW}) and was first considered by Razborov and Wigderson \cite{RW93}.  For a given row, a circuit size of $s$ and a correlation bound of $\alpha$ means that every size-$s$ circuit of the stated type agrees with the $n$-variable $\RW$ function on at most ${\frac 1 2} + \alpha$ fraction of inputs.   For each row, see Appendix~\ref{sec:NW} for a derivation of how the final column (seed length for a Nisan--Wigderson based PRG) follows from the earlier columns via \cite{NW94}.}
\end{table}

\subsection{Our main technical contribution:  New correlation bounds against $\{\SYM,\THR\} \circ \acz$ circuits}
The technical heart of our main result is a new exponential correlation bound against $\{\SYM,\THR\} \circ \acz$ circuits of size $n^{\Omega(\log n)}$: 
\begin{theorem} \label{thm:cor-bound}
There is an absolute constant $\tau>0$ and an explicit $\poly(n)$-time computable function $H : \zo^n \to \zo$ with the following property:  for any constant $d$, for $n$ sufficiently large it is the case that for any $n$-variable circuit $C$ of size $n^{\tau \log n}$ and depth $d$ with a $\SYM$ or $\THR$ gate at the top, we have 
\[
\Prx_{\bx \leftarrow \zo^n}[H(\bx) = C(\bx)] \leq {\frac 1 2} + \exp(-\Omega(n^{0.499})).
\]  
\end{theorem} 

Theorem~\ref{thm:cor-bound} strictly improves on the correlation bound provided by Theorem~4 of \cite{Vio06}, as it establishes correlation bounds for the same class of $n^{\Omega(\log n)}$-size circuits, but gives a much smaller $ \exp(- \Omega(n^{0.499}))$ upper bound on the correlation rather than $n^{-\Omega(\log n)}.$  As described in Appendix~\ref{sec:NW}, our  PRG result for $\{\SYM,\THR\} \circ \acz$ (Theorem~\ref{thm:main}) follows directly from Theorem~\ref{thm:cor-bound} via the Nisan--Wigderson framework.  In Section~\ref{sec:struc} we give an overview of the ideas that underlie our new correlation bound.

\pparagraph{Correlation bounds and PRGs for constant-depth circuits with multiple $\SYM$ or $\THR$ gates.} The main correlation bound and PRG of \cite{Vio06} are actually for $n^{c_d \log n}$-size depth-$d$ circuits with $c_d (\log n)^2$ many $\SYM$ gates, and similarly the main result of \cite{LS11} is a correlation bound for constant-depth circuits with $n^{1-o(1)}$ many $\SYM$ gates or $n^{1/2 - o(1)}$ many $\THR$ gates.  Our results similarly extend to constant-depth circuits with multiple $\SYM$ or $\THR$ gates. Our most general correlation bound is
the following: 

\begin{theorem} \label{thm:cor-bound-many-gates}
There is an absolute constant $\tau>0$ and an explicit $\poly(n)$-time computable function $H : \zo^n \to \zo$ with the following property:  for any constant $d$, for $n$ sufficiently large, any  $n$-variable circuit $C$ of size $n^{\tau \log n}$ and depth $d$ containing $n^{0.249}$ many $\SYM$ or $\THR$ gates (the circuit is allowed to contain both types of gates) satisfies
\[
\Pr_{\bx \leftarrow \zo^n}[H(\bx) = C(\bx)] \leq {\frac 1 2} + \exp(-\Omega(n^{0.249})).
\]  

\end{theorem}

Via the Nisan--Wigderson framework, Theorem~\ref{thm:cor-bound-many-gates} immediately yields the following, which is our most general PRG result: 

\begin{corollary} \label{cor:prg-many-gates} For some sufficiently small absolute constant $c>0$, there is a PRG with seed length $2^{O(\sqrt{\log S})} + \polylog(1/\eps)$ that $\eps$-fools the class of size-$S$ constant-depth circuits that contain $2^{\sqrt{c \log S}}$ many $\SYM$ or $\THR$ gates.
\end{corollary}

This strictly improves the main \cite{Vio06} PRG (Theorem~1 of \cite{Vio06}), which achieves seed length $2^{O(\sqrt{\log (S/\eps)})}$ for size-$S$ constant-depth circuits that contain $O((\log S)^2)$ many $\SYM$ or $\THR$ gates.  We prove Theorem~\ref{thm:cor-bound-many-gates} in Appendix \ref{sec:many-gates}.

\subsection{Barriers to further progress: correlation bounds for $\F_2$ polynomials and $\mathsf{ACC}^0$ lower bounds}
\label{sec:barriers} 
In this section we outline why achieving better dependence on $S$ will require groundbreaking new correlation bounds or circuit lower bounds. 

The seminal work of Babai, Nisan, and Szegedy~\cite{BNS92} gave an explicit function and established that it has exponentially small correlation $\exp(-\Omega(n/4^d \, d))$ with any $n$-variable $\F_2$ polynomials of degree $d$
(see Theorem~3 of \cite{Viola09now}). This result (and the multiparty communication-based techniques underlying it) have had far-reaching consequences in complexity theory; 25 years later, improving on this correlation bound remains a prominent open problem. In particular, even achieving correlation bounds of the form $\frac1{2} + n^{-1}$ against polynomials of degree $\log n$ with respect to any explicit distribution $\calD$ would constitute a significant breakthrough (see e.g.~``Open Question 1'' in Viola's excellent survey~\cite{Viola09now}).  Since every degree-$d$ polynomial is $s$-sparse for $s = {n \choose d}$, this is clearly a special case of obtaining $\frac1{2} + n^{-1}$ correlation bounds against polynomials of sparsity $s = {n \choose \log n}$.   Via a standard connection between PRGs and correlation bounds (see e.g.~Proposition 3.1 of~\cite{Vio09}),  an improvement in the dependence on $S$ in Theorem~\ref{thm:main} to $2^{o(\sqrt{\log S})} + \polylog(1/\eps)$, even for the special case of $S$-sparse $\F_2$ polynomials, would immediately yield exponentially-small correlation bounds against $s$-sparse $\F_2$ polynomials for $s = n^{\omega(\log n)} \gg {n\choose \log n}$ with respect to an explicit distribution. (We remark that the same correlation bounds are also open for the class of degree $\log n$ polynomial threshold functions, and hence the same barrier applies to improving the $S$-dependence of PRGs for size-$S$ $\THR \circ \AND$ circuits.)

Further improvements of Theorem~\ref{thm:main} would have even more dramatic consequences.  Classical ``depth-compression'' results of Yao~\cite{Yao:90} and Beigel and Tauri~\cite{BeigelTarui:94} (see also~\cite{CP16}) show that every size-$s$ depth-$d$ $\mathsf{ACC}^0$ circuit can be computed by a size-$S$ $\SYM \circ \AND$ circuit where $S = \exp((\log n)^{O_d(1)})$.  Improving the seed length of Theorem~\ref{thm:main} for the class of size-$S$ $\SYM \circ \AND$ circuits to $2^{(\log S)^{o(1)}}$ (even for constant $\eps$) would therefore separate $\mathsf{NP}$ from $\mathsf{ACC}^0$, a significant strengthening of Williams's celebrated separation of $\mathsf{NEXP}$ from $\mathsf{ACC}^0$~\cite{Williams11ccc}.

\subsection{The high-level structure of our correlation bound argument}
\label{sec:struc}

We recall the ``bottom-up'' approach to proving correlation bounds via the method of random restrictions.  This approach dates back to the classic correlation bounds between {\sc Parity} and $\acz$  of \cite{ajtai1983} and \cite{Hastad86}; in particular, the relevant prior works of~\cite{Vio06,LS11} also operate within this framework. 

Fix a hard function $H$, and let $F$ be any function belonging to a given class $\mathscr{F}$ of Boolean functions (in our case $\mathscr{F}$ is the class of $\{ \SYM,\THR\} \circ \acz$ circuits of size $n^{\Omega(\log n)}$). Our goal is to show that $F$ has small \emph{correlation} with $H$, i.e. that $\Pr_{\bx \leftarrow \zo^n}[F(\bx) = H(\bx)] \leq {\frac 1 2} + \alpha$ for some small $\alpha$ where $\bx$ is uniform over $\zo^n.$  This can be achieved by designing a  \emph{fair} distribution $\calR$\ignore{(which may depend on $F$)} over random restrictions that satisfies the following two competing requirements.  (A distribution $\calR$ over restrictions is said to be fair if first drawing a restriction $\brho \leftarrow \calR$ and then filling in all $\ast$'s to independent uniform values from $\zo$ results in a uniform random string from $\zo^n.$)  

\begin{enumerate}

\item [(1)] {\bf Approximator ($F$) simplifies:} With high probability $1-\gamma_{\SL}$ over $\brho\leftarrow\calR$, $F$ ``collapses" when it is hit by $\brho$, meaning that $F \uhr\brho \in \mathscr{F}_{\text{simple}}$ for some class $\mathscr{F}_{\text{simple}} \sse \mathscr{F}$. Looking ahead, in our case 
\[ \mathscr{F}_{\text{simple}} = \Big\{\text{$\{\SYM, \THR\} \circ \AND_{k := 0.0005\log m}$ circuits}\Big\}  \]
where $m \approx \sqrt{n}$ and $\AND_k$ denotes the class of fan-in $k$ $\AND$ gates. A collapse to this $\mathscr{F}_{\text{simple}}$ is useful for us because there are efficient multiparty communication protocols for functions computable by $\mathscr{F}_{\text{simple}}$ (due to \cite{HG91} when the top gate is $\SYM$ and to \cite{Nis93} when it is $\THR$).

\item [(2)] {\bf Target ($H$) retains structure:}  With high probability $1-\gamma_{\target}$ over $\brho\leftarrow\calR$, the restricted hard function $H \uhr \brho$ ``retains structure'', in the sense that it has small correlation with every function in $\mathscr{F}_{\text{simple}}$. In our case our notion of structure will be that $H \uhr\brho$ ``contains a perfect copy of'' the generalized inner product function: 
\[ \GIP_{m/2,k+1}(x) := \bigoplus_{i=1}^{m/2} \bigwedge_{j=1}^{k+1} x_{i,j},  \]
where $m$ and $k$ are the same $m$ and $k$ as above.
\end{enumerate} 

Suppose we have such a fair distribution $\calR$ over random restrictions satisfying (1) and (2) above.  The remaining step in the argument is to show the following:  (3) for any $\rho$ such that both of the above happen (approximator simplifies and the hard function retains structure), $F \uhr \rho$ and $H \uhr \rho$ have small correlation, i.e. they agree on at most $\frac1{2} + \gamma_{\corr}$ fraction of all inputs.  As in the previous works of \cite{Vio06,LS11}, the fact that $\mathscr{F}_{\text{simple}}$ and $\mathrm{GIP}_{m/2,k+1}$ have small correlation follows from  a celebrated theorem of Babai, Nisan, and Szegedy~\cite{BNS92} lower bounding the multiparty communication complexity of $\GIP_{m,k+1}$. 

It is straightforward to see that items (1)--(3) above establish a correlation bound of 
\[ \Prx_{\bx \leftarrow \zo^n}[\,F(\bx) = H(\bx)\,] \le \frac1{2} + \gamma_{\SL} + \gamma_{\target} + \gamma_{\corr}.\] 
The goal is therefore to carry out the above with $\max\{ \gamma_{\SL}, \gamma_{\target}, \gamma_{\corr} \}$ as small as possible for the class $\mathscr{F}$ of $\{\SYM,\THR\} \circ \acz_d$ circuits of as large a size as possible.  As indicated earlier, for any constant $d$ and circuits of size up to $s=n^{\tau \log n}$ we achieve $\max\{ \gamma_{\SL}, \gamma_{\target}, \gamma_{\corr} \}
= \exp(-\Omega(n^{0.499})).$  

After giving some technical preliminaries in Section~\ref{sec:prelim}, we upper bound
$\gamma_{\SL}$,
$\gamma_{\target}$, and
$\gamma_{\corr}$ 
in Sections~\ref{sec:bash},~\ref{sec:target}, and~\ref{sec:endgame} respectively.

\subsection{How this work differs from~\cite{Vio06,LS11}: improved depth reduction} 

A simple observation (due to~\cite{HM04}) that is used in both~\cite{Vio06,LS11} and in our work as well is the fact that a symmetric function of depth-$k$ decision trees can be simulated by a (different) symmetric function of width-$k$ $\AND$'s, and likewise for a threshold function of depth-$k$ decision trees. (See Fact~\ref{fact:fold} for a precise statement.)  Consequently we can think of $\mathscr{F}_{\text{simple}}$ as $\{ \SYM, \THR \} \circ \DT_k$ rather than $\{ \SYM, \THR\} \circ \AND_k$ (where $\DT_k$ denotes the class of decision trees of depth $k$), and for depth reduction it suffices to prove that a family of $s$ many $\acz$ circuits collapses to a family of small-depth decision trees with high probability under a random restriction. This is exactly what is shown by \emph{switching lemmas}. 

The loss in the previous works of~\cite{Vio06,LS11} is due to the switching lemmas they use and the limitations of these switching lemmas.  \cite{Vio06} uses the standard~\cite{Hastad86} switching lemma: 
 
 \begin{theorem}[H{\aa}stad's switching lemma] 
 \label{thm:HSL}
 Let $F$ be computed by a depth-$2$ circuit with bottom fan-in $w$. Then 
 \[ \Prx_{\brho\leftarrow\calR_p}[\,\text{$F\uhr\brho$ is not a depth-$t$ decision tree}\,\big] \le (5pw)^t.   \] 
 \end{theorem} 
 
 This failure probability of $(5pw)^t$ cannot be made exponentially small in our setting: since correlation bounds strong enough to be useful for the \cite{NW94} framework are not known for $\SYM \circ \AND_{\omega(\log n)}$ (see ``Open Question 1'' of \cite{Viola09now}) the value of $t$ has to has to be taken to  be at most $k = O(\log n)$, and moreover $p$ certainly has to be $\gg 1/n$ (since taking $p = 1/n$ would leave only a constant number of coordinates alive, and $H \uhr \brho$ would not ``retain structure'' in the sense of containing a copy of $\GIP_{m/2,k+1}$).  Indeed,~\cite{Vio06} applies Theorem~\ref{thm:HSL} with $p = n^{-\Theta(1)}$  \ignore{\rnote{Here it said ``--- harsher than standard applications of the HSL ---'' but this may confuse people who are familiar with applying it with $p=1/(10 \log S)$ in a setting where $S=2^{n^{1/d}}$ --- they may think of $p$ in such a setting as being $n^{-\Theta(1)}$ since $d$ is typically viewed as constant.  Okay to leave this ``--- harsher than standard applications of the HSL ---'' out?}}in order to make the failure probability as small as $n^{-\Omega(\log n)}$, and this is why~\cite{Vio06} only achieves quasi-polynomial correlation bounds $n^{-\Omega(\log n)}$. 
 
Faced with this obstacle, instead of using the standard~\cite{Hastad86} switching lemma, \cite{LS11} reverts to the earlier ``multi-switching lemma'' of \cite{ajtai1983} which applies to a collection of depth-2 circuits rather than a single such circuit.  The \cite{ajtai1983} multi-switching lemma, stated below, \emph{does} achieve exponentially small failure probability, but is only able to handle collections of $n^{O(\log\log n)}$ many $k$-DNFs, for $k = O(\log\log n)$.  Recall that a restriction tree $T$ is like a decision tree except that leaves do not have labels associated with them (so each root-to-leaf path is a restriction).  The distribution $\mu_T$ corresponds to the distribution over restrictions obtained by making a random walk from the root of $T$. 

\begin{theorem}[Ajtai's switching lemma \cite{ajtai1983}] \label{thm:ajtaiSL}
Let $\calF = \{F_1,\dots,F_{s}\}$ be a family of $s$ many DNFs over $x_1,\dots,x_n$, each of width $k$.  For any $t \geq 1$, there is a restriction tree $T$ of height at most $nk (\log s)/(\log n)^t$ such that
\[
\Prx_{\brho \leftarrow \mu_T}\big[F_i \uhr \brho \text{~is not a $(\log n)^{10kt2^k}$-junta}\big] \leq
2^{-n/(2^{10k}(\log n)^t)}.
\]
\end{theorem}
Hence~\cite{LS11} achieves exponentially small correlation bounds (the main point of their paper), but only against circuits of size $n^{O(\log\log n)}$. 

The key new ingredient that we employ in this work is a recent powerful multi-switching lemma from \cite{Has14}.  (We note that \cite{IMP12} gives an essentially equivalent multi-switching lemma which we could also use.)  Roughly speaking the \cite{Has14} multi-switching lemma, whose precise statement we defer to Section~\ref{sec:bash} as it is somewhat involved, lets us achieve an exponentially small failure probability (like Ajtai's multi-switching lemma) of achieving a significantly more drastic simplification than Ajtai's multi-switching lemma (recall the doubly-exponential-in-$k$ dependence on the junta size in Theorem~\ref{thm:ajtaiSL}).  This quantitative improvement in depth reduction translates into our stronger correlation bounds.

\subsection{Relation to \cite{ST18:improved}}
We close this introduction by discussing the connection between this paper and recent concurrent work of the authors \cite{ST18:improved}.  The high-level approaches of the two paper are fairly different: unlike the current paper, \cite{ST18:improved} does not use the Nisan--Wigderson hardness-versus-randomness paradigm (and does not establish any new correlation bounds); instead it establishes a derandomized version of the \cite{Has14} multi-switching lemma and combines this with other ingredients to obtain its final PRG in a manner reminiscent of \cite{AjtaiWigderson:85,TX13}.

The results of the two papers are also incomparable (briefly,~\cite{ST18:improved} obtains significantly shorter seed length for significantly more restricted classes of functions).  The first main result of \cite{ST18:improved} is an $\eps$-PRG for the class of size-$S$ depth-$d$ $\acz$ circuits with seed length $\log(S)^{d+O(1)} \cdot \log(1/\eps)$.  This is incomparable to the most closely related result of the present paper (Corollary~\ref{cor:prg-many-gates}, which gives a $2^{O(\sqrt{\log S})} + \polylog(1/\eps)$ seed length PRG for $\acz$ circuits augmented with polynomially many $\SYM$ or $\THR$ gates), since the \cite{ST18:improved} result gives a significantly better seed length but for the significantly more limited class of ``un-augmented'' constant-depth circuits (indeed, the~\cite{ST18:improved} result does not apply to $\acz$ circuits augmented even with a single $\SYM$ or $\THR$ gate).  The second main result of \cite{ST18:improved} is an $\eps$-PRG for the class of $S$-sparse $\F_2$ polynomials with seed length $2^{O(\sqrt{\log S})} \cdot \log(1/\eps).$  Here too the seed length of \cite{ST18:improved} is shorter than that of the current paper (giving the optimal $\log(1/\eps)$ dependence on $\eps$ as opposed to the $(\log(1/\eps))^{4.01}$ of the current paper), but the result of \cite{ST18:improved} only holds for $S$-sparse $\F_2$ polynomials, which are a very restricted case of the $\{\SYM,\THR\} \circ \acz_d$ circuits which are handled in the current paper.

\section{Preliminaries} \label{sec:prelim}

We use bold font like $\bx$, $\brho$, etc. to denote random variables.

We write ``size-$S$ $\acz_d$'' to denote the class of circuits of depth $d$ consisting of at most $S$ unbounded fan-in $\AND/\OR$ gates with variables and negated variables as the inputs (we include these literals in the gate count).

\pparagraph{Pseudorandomness.}  For $r < n$, we say that a distribution $\calD$ over $\zo^n$ can be \emph{sampled efficiently with $r$ random bits} if (i) $\calD$ is the uniform distribution over a multiset of size exactly $2^r$ of strings from $\zo^n$, and (ii) there is a deterministic algorithm $\mathrm{Gen}_{\calD}$ which, given as input a uniform random $r$-bit string $\bx \leftarrow \zo^r$, runs in time $\poly(n)$ and outputs a string drawn from $\calD$.

For $\delta>0$ and a class $\calC$ of functions from $\zo^n$ to $\zo$, we say that a distribution $\calD$ over $\zo^n$ \emph{$\delta$-fools $\calC$ with seed length $r$} if (a) $\calD$ can be sampled efficiently with $r$ random bits via algorithm $\mathrm{Gen}_\calD$, and (b) for every function $f \in \calC$, we have
\[
\bigg|\Ex_{\bs \leftarrow \{0,1\}^r}[f(\mathrm{Gen}_{\calD}(\bs))] - 
\Ex_{\bx \leftarrow \{0,1\}^n}[f(\bx)]\bigg| \leq \delta.
\]
Equivalently, we say that $\mathrm{Gen}_\calD$ is a \emph{$\delta$-PRG for $\calC$ with seed length $r$.}

\ignore{
}

\pparagraph{Restrictions.}  A \emph{restriction} $\rho$ of variables $x_1,\dots,x_n$ is an element of $\{0,1,\ast\}^n$).
Given a function $f(x_1,\dots,x_n)$ and a restriction $\rho$, we write $f \uhr \rho$ to denote the function obtained by fixing $x_i$ to $\rho(i)$ if $\rho(i) \in \{0,1\}$ and leaving $x_i$ unset if $\rho(i)=\ast.$
 For two restrictions $\rho,\rho'\in \{0,1,\ast\}^{n}$, their \emph{composition}, denoted $\rho\rho'\in \{0,1,\ast\}^{n}$, is the restriction defined by
\[ (\rho\rho')_i = \left\{
\begin{array}{cl}
\rho_i & \text{if $\rho_i \in \{0,1\}$} \\
\rho'_i & \text{otherwise.}
\end{array}
\right.\]
We write $\calR_p$ to denote the standard distribution over random restrictions with $\ast$-probability $p$, i.e. $\brho$ drawn from $\calR_p$ is a random string in $\{0,1,\ast\}$ obtained by independently setting each coordinate to $\ast$ with probability $p$ and to each of $0,1$ with probability ${\frac {1-p} 2}.$

\ignore{
}

\subsection{Multiparty communication complexity} 
\label{sec:BNS} 

We recall a celebrated lower bound of Babai, Nisan, and Szegedy \cite{BNS92} on the multi-party ``number on forehead'' (NOF) communication complexity of the generalized inner product function:

\begin{theorem}[\cite{BNS92}] \label{thm:BNS}
\ignore{\rnote{This version of the statement is from \cite{LS11}, right?  \cite{Vio06} states a version that doesn't have the $\gamma_{\err}$ parameter. As a nervous pedant I would like to understand and maybe state the statement more clearly:  $P$ is a randomized NOF protocol and its output agrees with $f$ w.p. $(1-\gamma_{\err})$ on each input (i.e. it's the exact same communication complexity model as the Nisan theorem, Theorem \ref{thm:nisan-PTF} below, correct?}}
There is a partition of the $m\cdot (k+1)$ inputs of 
\[ \GIP_{m,k+1}(x) := \bigoplus_{i=1}^m \bigwedge_{j=1}^{k+1} x_{i,j}   \]
into $k+1$ blocks such that the following holds: Let $P$ be a $(k+1)$-party randomized NOF communication protocol exchanging at most $\frac1{10}(m/4^{k+1} - \log(1/\gamma_{\comm}))$ bits of communication and computing a Boolean function $f$ with error $\gamma_{\err}$ (meaning that on every input $x$ the protocol outputs the correct value $f(x)$ with probability at least $1-\gamma_{\err}$). Then 
\[ \Prx_{\bx\leftarrow \zo^{m(k+1)}}\big[f(\bx) = \GIP_{m,k+1}(\bx)\big] \le \frac1{2} + \gamma_{\err} + \gamma_{\comm}.\] 
\end{theorem}

The connection between $\SYM\circ\mathsf{AND}_k$ circuits and $(k+1)$-party communication complexity is due to the following simple but influential observation of H{\aa}stad and Goldmann:

\begin{fact}[\cite{HG91}]
\label{fact:HG}
Let $f : \zo^n\to\zo$ be a Boolean function computed by a size-$s$ $\SYM\circ\AND_k$ circuit. Then for any partition of the $n$ inputs of $f$ into $k+1$ blocks, there is a deterministic NOF $(k+1)$-party communication protocol that computes $f$ using $O(k\log s)$ bits of communication.
\end{fact}

For $\THR \circ \acz$ circuits we use an analogous result from~\cite{Nis93} on the $(k+1)$-party randomized $\gamma$-error communication complexity of $\THR \circ \AND_k$ circuits:

\begin{theorem}[\cite{Nis93}] \label{thm:nisan-PTF}
Let $f : \zo^n\to\zo$ be a Boolean function computed by a $\THR\circ\AND_k$ circuit. Then for any partition of the $n$ inputs of $f$ into $k+1$ blocks, there is a randomized NOF $(k+1)$-party communication protocol that computes $f$ with error $\gamma_{\err}$ using $O(k^3\log n\log(n/\gamma_{\err}))$ bits of communication. 
\end{theorem}

\section{Ingredient (1):  Simplifying the approximator}\ignore{: \cite{Has14} multi-switching lemma}
\label{sec:bash}

The main result of this section is  the following:

\begin{lemma} \label{lem:fair-1}
Let $F$ be any $\{\SYM,\THR\} \circ \acz_d$ circuit of size $s=n^{\tau \log n}$.
There is a fair distribution $\calR$ over restrictions $\brho \in \{0,1,\ast\}^n$ such that the following holds:  With probability
$1-\gamma_{\SL} = 1-\exp(-\Omega_d(\sqrt{n/\log n}))$ over the draw of $\brho \leftarrow \calR$,
it is the case that $F \uhr \brho$ belongs to the class $\mathscr{F}_{\text{simple}} = \{\SYM, \THR\} \circ \AND_{k := 0.0005\log m}$.  
\end{lemma}

The recent ``multi-switching lemma'' of \cite{Has14} is the main technical tool we use to establish Lemma~\ref{lem:fair-1}.  To state the \cite{Has14} lemma we need some terminology.  Let $\mathscr{G}$ be a family of Boolean functions.  A restriction tree $T$ is said to be a \emph{common $\ell$-partial restriction tree (RT) for $\mathscr{G}$} if every $g \in \mathscr{G}$ can be expressed as 
$T$ with depth-$\ell$ decision trees hanging off its leaves. (Equivalently, for every $g \in \mathscr{G}$ and root-to-leaf path $\pi$ in $T$, we have that $g \uhr \pi$ is computed by a depth-$\ell$ decision tree.)

\begin{theorem}[\cite{Has14} multi-switching lemma]
\label{thm:H14-SL}
Let $\mathscr{F} = \{ F_1,\ldots,F_s\}$ be a collection of depth-$2$ circuits with bottom fan-in $w$.  Then for any $t \geq 1$,
\[ \Prx_{\brho'\leftarrow\calR_p} \big[\,\mathscr{F}\uhr\brho' \text{~does not have a common $(\log s)$-partial RT of depth $\le t$~} \big] \le s(24pw)^t.\]
\end{theorem}

Theorem~\ref{thm:H14-SL} is the main tool we use to simplify any $\{\SYM,\THR\} \circ \acz$ circuit down to an $\mathscr{F}_{\text{simple}}$-circuit.  Conceptually, we think of this transformation as being done in three steps: 

\begin{enumerate}
\item (Main step) Apply a random restriction $\brho' \leftarrow \calR_{p}$ to convert a $\{\SYM,\THR\} \circ \acz$ circuit into a decision tree with a $\{\SYM,\THR\} \circ \DT$ circuit at each leaf.
\item Observing that $\{\SYM,\THR\} \circ \DT \equiv \{\SYM,\THR\} \circ \AND$, this is equivalent to a decision tree with a $\{\SYM,\THR\} \circ \AND$ circuit at each leaf.
\item Trim the fan-in of the $\AND$ gates in each $\{\SYM,\THR\} \circ \AND$ circuit (by increasing the depth of the decision tree).  The last step in the draw of a random restriction $\brho$ from the overall fair distribution $\calR$ corresponds to a random walk down this final decision tree.
\end{enumerate} 


In the rest of this section we describe each of these steps in detail and thereby prove Lemma~\ref{lem:fair-1}.

\pparagraph{First (main) step.}
If $g$ is a Boolean function and $\calC$ is a class of circuits, we say that
$g$ is \emph{computed by a $(d,\calC)$-decision tree} if $g$ is computed by a decision tree of depth $d$ (with a single Boolean variable at each internal node as usual) in which each leaf is labeled by a function from $\calC.$  We require the following corollary of Theorem~\ref{thm:H14-SL}:

\begin{corollary}
\label{cor:bash} Let $G$ be any Boolean function and $\G$ be a gate computing $G$, and let $F$ be a  $\G \circ \acz_d$ circuit of size $s$ . Then for $p = \frac1{48}(48\log s)^{-(d-1)}$ and any $t \geq 1$,
\[ \Prx_{\brho'\leftarrow\calR_{p}}\big[\,F\uhr \brho' \text{~is not computed by a $(2^dt,\G \circ \DT_{\log s})$-decision tree}\,\big]  \le s \cdot 2^{-t}.\] 
\end{corollary}

\begin{proof} 
We may assume without loss of generality that the depth-$(d+1)$ circuit $F$ is \emph{layered}, meaning that for
any gate $g$ it contains, every directed path from an input variable to $g$ has the same length (converting an unlayered circuit to a layered one increases its size only by a factor of $d$, which is negligible for our purposes). Let $s_i$ denote the number of gates in layer $i$ (at distance $i$ from the inputs), so $s = s_1 + \cdots + s_d$.

We begin by trimming the bottom fan-in of $F$: applying Theorem~\ref{thm:H14-SL} with $\mathscr{F}$ being the $s_1$ many bottom layer gates of $F$ (viewed as depth-$2$ circuits of bottom fan-in $w=1$) and $p_0 := 1/48$, we get that 
\[ \Prx_{\brho_0\leftarrow\calR_{p_0}}\big[\, F\uhr\brho_0\text{~is not computed by a $(t,G \circ \acz$(depth $d$, bottom fan-in $\log s))$-decision tree}\,\big] \le  s_1 \cdot 2^{-t}.\]

Let $F^{(0)}$ be any good outcome of the above, a $(t,G \circ \acz$(depth $d$, bottom fan-in $\log s))$-decision tree. 
Note that there are at most $2^t$ many $\acz(\text{depth $d$, fan-in $\log s$})$ circuits at the leaves of the depth-$t$ decision tree.  Applying Theorem~\ref{thm:H14-SL} to each of them with $p_1 := 1/(48\log s)$ (and the `$t$' of Theorem~\ref{thm:H14-SL} being $2t$) and taking a union bound over all $2^t$ many of them, we get that 
\begin{align*} &\Prx_{\brho_1\leftarrow\calR_{p_1}}\big[\, F^{(0)}\uhr\brho_1\text{~is not a $(t+2t,G \circ \acz$(depth $d-1$, fan-in $\log s))$-decision tree}\,\big]  \\
&\le s_2 \cdot 2^{-2t} \cdot 2^t = s_2 \cdot 2^{-t}. 
\end{align*}
Repeat with $p_2 = \ldots = p_{d-1} := 1/(48\log s)$, each time invoking Theorem~\ref{thm:H14-SL} with its `$t$' being the one more than the current depth of the decision tree . The claim then follows by summing the $s_1 2^{-t}$, $s_2 2^{-t}, \dots,
s_d 2^{-t}$ failure probabilities over all $d$ stages and the fact that 
\[ \prod_{j=0}^{d-1} p_i = \frac1{48}\cdot \frac1{(48\log s)^{d-1}} = p.\] 
\vskip -.4in\end{proof} 

\bigskip

\pparagraph{Second step:  From $\{\SYM,\THR\} \circ \DT$ to $\{\SYM,\THR\} \circ \AND$.}  We recall the following fact from~\cite{HM04}: 

\begin{fact}
\label{fact:fold} 
Every $\SYM_s \circ \DT_{\log s}$ function (resp.~$\THR_s \circ \DT_{\log s}$) can be computed by a $\SYM_{s^2} \circ \AND_{\log s}$ (resp.~$\THR_{s^2} \circ\AND_{\log s}$) circuit. 
\end{fact} 

(This is an easy consequence of the fact that any decision tree may be viewed as a DNF whose terms corresponds to the paths to 1-leaves, and that this DNF has the property that any input assignment makes at most one term true.)  Applying Fact~\ref{fact:fold} and choosing $t = m/2^{d+1}$ in Corollary~\ref{cor:bash}, (where $m = \Theta(\sqrt{n/\log n})$ will be defined precisely in the next section), we get the following special case of Corollary~\ref{cor:bash}: 

\begin{corollary} \label{cor:collapse}
Let $F$ be a $\{\SYM,\THR\} \circ \acz_d$ circuit of size $s =  n^{\tau \log n}$. Then for $p = \frac1{48}(48\log s)^{-(d-1)}$,
\ignore{\rnote{The first line of the math was
\[ 
\Prx_{\brho\leftarrow\calR_p}\big[\,F\uhr \brho \text{~is not computed by a $(m/2,\{\SYM_{s^2},\THR_{s^2}\}\circ \red{\DT_{\log s}})$-decision tree}\,\big]
\]
but the $\DT$ should be $\AND$, right?}}
\begin{align} 
& \Prx_{\brho'\leftarrow\calR_p}\big[\,F\uhr \brho' \text{~is not computed by a $(m/2,\{\SYM_{s^2},\THR_{s^2}\}\circ \AND_{\log s})$-decision tree}\,\big] \nonumber \\
& \le  s \cdot 2^{-t} = s \cdot 2^{-m/2^{d+2}} \nonumber \\
&= \exp(-\Omega_d(\sqrt{n/\log n})) \ := \ \gamma_{\SL}. \label{eq:gamma-SL}
\end{align}
\end{corollary}

\pparagraph{Third step:  Trimming to reduce bottom fan-in.}  
The $\{ \SYM,\THR\} \circ \AND$ circuits hanging off the leaves of our decision tree have bottom fan-in at most $\log s$, but we will need them to have fan-in at most $k$ in order to invoke the~\cite{BNS92} lower bound later. At each leaf $\ell$ we achieve this smaller fan-in by  identifying a set (call it $S_\ell$) of additional variables and restricting them in all possible ways; we argue that every fixing of the variables in $S_\ell$ gives the desired upper bound of $k$ on the bottom-$\AND$ fan-in. We use a probabilistic argument to establish the existence of the desired set $S_\ell$ (this is important because in the next section we will need each $S_\ell$ to satisfy an additional property, and the probabilistic argument makes it easy to achieve this).

 Let us write ``$\bL \sse_q X$'' to indicate that $\bL$ is a subset of $X$ that is randomly chosen by independently including each element of  $X$ with probability $q$.  We will use the following easy result:
\begin{fact}
\label{fact:trim}
Let $\{ C_1,\ldots,C_{s^2}\}$ be a collection of subsets of $[n]$ where each $|C_i| \leq w.$ Then for $\bL \sse_q [n]$ and $k \le w$, we have
\begin{align*} \Prx_{\bL \sse_q [n]}\big[\,\exists\, i \in [s^2] \text{~such that $|C_i \cap \bL| > k$} \,\big] &\le s^2 {w\choose k} q^k. \end{align*}
\end{fact} 
 
Recall that $s = n^{\tau \log n}$ where $\tau > 0$ is a small absolute constant to be specified later and that $k = 0.0005\log n$.
We set 
\[ q := \frac{k}{e\log s}\cdot 2^{-(3\log s)/k} = \frac1{\Theta(\log n)} \cdot n^{-\Theta(1) \cdot \tau}< n^{-0.01}, \]
where the last inequality holds for a suitably small choice of the constant $\tau$. 
Observe that $q$ is chosen so as to ensure  
\begin{equation}
s^2 {\log s \choose k}q^k \le 2^{2\log s} \left(\frac{e\log s}{k}\cdot  q \right)^k = \frac1{s} \ll 1.
\label{eq:zero-point-one}
\end{equation}

Fix $T$ to be an $(m/2,\{\SYM_{s^2},\THR_{s^2}\}\circ \AND_{\log s})$-decision tree as given by Corollary~\ref{cor:collapse}.  At each leaf $\ell$ of $T$, draw a set $\bL(\ell) \sse_q [n]$ and let $\bS_\ell$ be $([n] \setminus \fixed(\ell)) \setminus \bL(\ell)$, where
$\fixed(\ell) \subseteq [n]$ is the subset of variables that are fixed on the root-to-$\ell$ path in $T$. By Fact~\ref{fact:trim} and (\ref{eq:zero-point-one}), at each leaf $\ell$ it is the case that with probability at least $1-1/s$ over the random draw of $\bL(\ell)$, every extension of the root-to-$\ell$ path in $T$ that additionally fixes all the variables in $\bS_\ell$ collapses the $\{\SYM,\THR\}\circ \AND_{\log s}$ circuit that was at $\ell$ in $T$ down to a $\{\SYM,\THR\}\circ \AND_{k}$ circuit.  We say that such an outcome of $\bL(\ell)$ is a \emph{good} outcome (we will refer back to this notion in the next section).

In summary, the above discussion establishes Lemma \ref{lem:fair-1}, where the fair distribution $\calR$ corresponds to 

\begin{itemize}

\item [(a)] first drawing $\brho' \leftarrow \calR_{p}$, 

\item [(b)] then walking down a random root-to-leaf path $\pi$ in the resulting depth-$(m/2)$ decision tree given by Corollary~\ref{cor:collapse}, 

\item [(c)] and then finally, at the resulting leaf $\ell$, choosing a random assignment to the variables in the set $S_\ell$ that corresponds to $L(\ell)$, where $L(\ell)$ is a good outcome of the random variable $\bL(\ell) \subseteq_q [n].$  (Note that the randomness over $\bL(\ell)$ is not part of the random draw of $\brho \leftarrow \calR$; all we require is the existence of a good $L(\ell).$)

\end{itemize}

Based on our discussion thus far each $L(\ell)$ may be fixed to be any good outcome of $\bL(\ell)$; we will give an additional stipulation on $L(\ell)$ in Remark \ref{rem:stipulation}.

%
%
%
%

\section{Ingredient (2) (target retains structure): $\GIP\circ \mathsf{PAR}$ under random restrictions}
\label{sec:target}

Like~\cite{Vio06,LS11}, our hard function will be the generalized inner product function composed with parity:
\begin{equation} \RW_{m,k,r}(x) = \bigoplus_{i=1}^m \bigwedge_{j=1}^{k+1} \bigoplus_{\ell=1}^r x_{i,j,\ell}.
\label{eq:RW}
\end{equation}
This function was introduced by Razborov and Wigderson~\cite{RW93} to show $n^{\Omega(\log n)}$ lower bounds against depth-3 threshold circuits with $\AND$ gates at the bottom layer.  We will set 
\[ m = r = \sqrt{n/(k+1)} \quad \text{(recall that~}k = 0.0005\log m\text{)}. \] 
Note that $m= r = \Theta(\sqrt{n/\log n})$ and $k = \Theta(\log n)$. Given parameters $m',k',r'$, we say that a function $g: \zo^n \to \zo$ \emph{contains a perfect copy of $\RW_{m',k',r'}$} if there is a restriction $\kappa$ such that $(g \uhr \kappa)(x) = b \oplus \bigoplus_{i=1}^{m'} \bigwedge_{j=1}^{k+1} \left( b_{i,j} \bigoplus_{\ell=1}^{r'} x_{i,j,k}\right)$ for some bits $b,b_{i,j}.$

Roughly speaking, the motivation behind augmenting $\GIP$ with a layer of parities is to ensure that $\RW$ is resilient to random restrictions (i.e.~that $\RW\uhr\brho$ ``remains complex'', containing a copy of $\GIP$ with high probability after a suitable random restriction). In our setting we need that $\RW$ is resilient to a random restriction $\brho \leftarrow \calR$ for the fair distribution $\calR$ from Lemma~\ref{lem:fair-1}; we establish this in the rest of this section.\ignore{with $*$-probability $p = 1/(\log n)^{\Theta(d)}$, i.e. the random restriction in Corollary~\ref{cor:collapse}; this is given by Proposition \ref{prop:chernoff} below.}  

\begin{proposition}
\label{prop:chernoff} 
Consider the space of formal variables of $\RW_{m,k,r} : \zo^n \to \zo$: 
\[ X = \big\{ x_{i,j,t} \colon (i,j,t) \in [m] \times [k+1] \times [r]\big\}, \quad |X| = m(k+1)r := n. \] 
Then for $p = \frac1{48}(48\log s)^{-(d-1)}$ (as in Corollary~\ref{cor:collapse}),
\begin{align*}
 \Prx_{\bL \sse_p X}\bigg[\, \exists\, (i,j) \colon \big| \{ t \in [r] \colon x_{i,j,t} \in \bL\}\big| < \frac{pr}{2} \bigg] &\le  m(k+1) \cdot \exp\left(-\Omega(pr)\right).
 \end{align*}
\end{proposition}

\begin{proof}
This follows directly from a standard multiplicative Chernoff bound\ignore{(Fact~\ref{fact:chernoff} with $c = 1/2$)} and a union bound over all $(i,j) \in [m]\times [k+1]$.  
\end{proof}

Recall that a random restriction $\brho'\leftarrow\calR_p$ can be thought of as being sampled by first drawing $\bK \sse_p X$ and setting $\brho_i$ to $\ast$ for each $i \in \bK$, and then setting the coordinates of $\brho'$ in $X \setminus \bK$ according to a uniform random draw from $\zo^{X\setminus \bK}.$  Proposition~\ref{prop:chernoff} and the definition of $\RW$ thus yield the following:

\begin{corollary} \label{cor:copy}
For $\brho' \leftarrow \calR_p$, for $p = \frac1{48}(48\log s)^{-(d-1)}$,
 $\RW_{m,k,r}(x) \uhr \brho'$ contains a perfect copy of
\begin{equation*} \RW_{m,k,r'}(x) = \bigoplus_{i=1}^m \bigwedge_{j=1}^{k+1} \bigoplus_{t=1}^{r'} x_{i,j,t}, \qquad \text{where~~} r' = \frac{pr}{2} \ignore{\text{~~and~ $p = \frac1{(\Theta(\log s))^{d-1}}$}}  
\end{equation*}
with failure probability at most 
\begin{equation}\exp(-\Omega(pr)) = \exp\left(-\frac{\sqrt{n/\log n}}{(\Theta(\log s))^{d-1}}\right) := \gamma_{\target}. \label{eq:gamma-target}
\end{equation}
\end{corollary}
Note that 
\[ r' = \frac{pr}{2} = \frac{\sqrt{n/\log n}}{(\Theta(\log s))^{d-1}} >n^{0.49},\] 
where the inequality uses the fact that $d$ is a constant and the fact that $s = n^{O(\log n)}$; we will use this later.

Corollary \ref{cor:copy} states that with very high probability over $\brho' \leftarrow \calR_p$, the function $\RW_{m,k,r} \uhr \brho'$ ``does not simplify too much''; however we need $\RW_{m,k,r}$ to ``not simplify too much'' under a full random restriction drawn from $\calR$ (recall the discussion at the end of Section~\ref{sec:bash}).  We proceed to establish this.

Fix any outcome $\rho'$ of $\brho' \leftarrow \calR_p$ such that (i) the conclusion of Corollary~\ref{cor:collapse} holds (i.e. $F \uhr \rho'$ is computed by a $(m/2,\{\SYM_{s^2},\THR_{s^2}\}\circ \DT_{\log s})$-decision tree, which we call $T$), and (ii) the conclusion of Corollary~\ref{cor:copy} holds (i.e. $\RW_{m,k,r} \uhr \rho'$ contains a perfect copy of $\RW_{m,k,r'}$).  (A random $\brho' \leftarrow \calR_p$ is such an outcome with probability at least $1 - \gamma_{\SL} - \gamma_{\target}$.)  For ease of notation let us write $\RW'$ to denote $\RW_{m,k,r} \uhr \rho'$.  

Fix any path $\pi$ that reaches a leaf $\ell$ in $T.$  (Note that a random choice of such a path corresponds to part (b) in the random draw of $\brho \leftarrow \calR$, recalling the discussion at the end of Section~\ref{sec:bash}.)  Since $|\pi| \le m/2$, we have that the set 
\[ A_\ell := \{ i \in [m] \colon \pi_{i,j,t} = \ast \text{~for all $j \in [k+1]$ and all $t$}\} \]  
has cardinality at least $m-|\pi| \ge m/2$. In words, at least $m/2$ of the $m$ many depth-$2$ subcircuits of $\RW'$ are completely ``untouched'' by $\pi$.  For part (c) of the draw from $\calR$, recall that the set $L(\ell)$ could be taken to be any good outcome of $\bL(\ell)$, and that a random $\bL(\ell) \subseteq_q [n]$ is good with probability at least $1-1/s.$  By the same Chernoff bound argument as the one in Proposition~\ref{prop:chernoff}, we have that 
\begin{align*} \Prx_{\bL\sse_q [n]}\bigg[\, \exists\, (i,j) \in A_\ell \times [k+1] \colon \big| \{ t  \colon x_{i,j,t} \in \bL(\ell)\}\big| < \frac{qr'}{2} \bigg] &\le  |A_\ell| (k+1)  \exp\left(-\Omega(qr')\right)\\
&\ll \exp(-\Omega(n^{0.48})),
\end{align*}
recalling that $|A_\ell| \leq m,$ $k = \Theta(\log n)$, $q \geq n^{-0.01}$ and $r' > n^{0.49}.$  Since $1-1/s + 1 - \exp(-\Omega(n^{0.48})) > 1$,   
there must exist a good outcome $L(\ell)$ of $\bL(\ell)$ such that for the corresponding $S_\ell$, every restriction $\rho^{\text{trim}}$ fixing precisely the variables in $S_\ell$ is such that $\RW' \uhr \pi \rho^{\text{trim}}$ contains a perfect copy of 
\begin{equation*} \RW_{m,k,r''}(x) = \bigoplus_{i \in A_\ell} \bigwedge_{j=1}^{k+1} \bigoplus_{t=1}^{r''} x_{i,j,t}, \qquad \text{where~~} r'' = \frac{qr'}{2} \gg 1. 
\end{equation*}
Having $r'' \geq 1$ is crucial for us because, together with $|A_\ell| \geq m/2$, it means that $\RW_{m,k,r''}$ contains a perfect copy of $\GIP_{m/2,k+1}$ (i.e. by possibly restricting and renaming some variables of $\RW_{m,k,r''}$ and possibly negating the result, we obtain a function identical to $\GIP_{m/2,k+1}$).  

\begin{remark} \label{rem:stipulation}
We refine the definition of $\calR$ to require that in (c) it use an $L(\ell)$ as specified above at each leaf $\ell$.
\end{remark}

Summarizing, the above discussion establishes that $\RW_{m,k,r}$ ``retains structure'' with high probability under a random $\brho \leftarrow \calR$. The formal statement of this result (incorporating also Lemma~\ref{lem:fair-1}) is as follows:
  
\begin{lemma} \label{lem:retain-structure}
Let $F$ be any $\{\SYM,\THR\} \circ \acz_d$ circuit of size $s=n^{\tau \log n}$.
The fair distribution $\calR$ over restrictions $\brho \in \{0,1,\ast\}^n$ from Lemma~\ref{lem:fair-1} satisfies the following:  With probability 
$1-\gamma_{\SL} - \gamma_{\target}$ over a draw of $\brho \leftarrow \calR$, both of the following hold:

\begin{enumerate}

\item [(i)] $F \uhr \brho$ belongs to  $\{\SYM, \THR\} \circ \AND_k$;

\item [(ii)] $\RW_{m,k,r} \uhr \brho$ contains a perfect copy of $\GIP_{m/2,k+1}$.

\end{enumerate}

\end{lemma}

\section{Bounding the correlation between the approximator and target post-restriction} 
\label{sec:endgame} 

With Lemma~\ref{lem:retain-structure} in hand it is a simple matter to finish the argument.  Fix any outcome $\rho$ of $\brho \leftarrow \calR$ such that $F \uhr \rho$ and $\RW_{m,k,r} \uhr \rho$ satisfy (i) and (ii) of Lemma~\ref{lem:retain-structure}.
Applying either Fact~\ref{fact:HG} or Theorem~\ref{thm:nisan-PTF} (depending on whether the top gate of $F$ is $\SYM$ or $\THR$ along with the lower bound of Theorem~\ref{thm:BNS}), we get that
\begin{equation}
\Prx_{\bx \leftarrow \zo^n}[\,(F \uhr \rho)(\bx) = (\RW_{m,k,r} \uhr \rho)(\bx)\,] \le
 \frac1{2} + \exp\left(-\Omega(m/4^k)\right) = 1/2 + \gamma_{\corr}, \label{eq:gamma-corr}
\end{equation}
where
\[
\gamma_{\corr} = \exp\left(-\Omega(m/4^k)\right) = \exp \left(-\Omega(m^{0.999})\right) = \exp\left(-\Omega(n^{0.499})\right).
\]
This gives ingredient (3) as described in Section~\ref{sec:struc}.  Recalling the discussion at the start of Section \ref{sec:struc},
Theorem~\ref{thm:cor-bound} follows from Lemma \ref{lem:retain-structure} and (\ref{eq:gamma-corr}).

\bibliography{allrefs}{}
\bibliographystyle{alpha}

\appendix

\section{Applying the \cite{NW94} paradigm to obtain pseudorandom generators from correlation bounds}
\label{sec:NW}

A function $f$ is said to be \emph{$(s,\tau)$-hard for a circuit class $\calC$} if every circuit $C \in \calC$ of size at most $s$ has $\Pr_{\bx}[f(\bx)=C(\bx)] \leq {\frac 1 2} + \tau$, where $\bx$ is a uniform random input string.  If this holds then we say that $f$  gives a \emph{correlation bound} of $\tau$ against $\calC$-circuits of size $s$.

Given a quadruple $(m,r,\ell,s)$ of non-negative integers, a family ${\cal F}=\{T_1,\dots,T_s\}$ of $r$-element subsets of $[m]$ is said to be an \emph{$(m,r,\ell,s)$-design} if for any two distinct subsets $T_i,T_j \in \calF$ we have $|T_i \cap T_j| \leq \ell.$

An $\ANY_t$ gate is a gate that takes in $t$ inputs and computes an arbitrary function from $\zo^t$ to $\zo$.

We recall the Nisan-Wigderson \cite{NW94} translation from correlation bounds to PRGs:
\ignore{Our statement of the NW generator is cribbed from \cite{HS16}.}

\begin{theorem} 
[The Nisan-Wigderson generator] \label{thm:NW} \ignore{ \rnote{This is a somewhat vague statement since we're not explaining what either ``explicit'' means.  If we want to be more precise about the running time of the generator then we should clarify this}} Fix a circuit class $\calC$ and let $m,r,\ell,s \in \N$ be positive parameters with $m \geq r \geq \ell$.  Given an explicit $f: \{0,1\}^r \to \{0,1\}$ that is $(s \cdot 2^\ell,
\eps/s)$-hard for $\calC \circ \ANY_{\log \ell}$ and an explicit $(m,r,\ell,s)$-design, there is an explicit PRG $G:  \zo^m \to \zo^s$ that $\eps$-fools size-$s$ circuits in $\calC$.  (Hence for $s \geq n$, by taking the first $n$ output bits of $G$ there is an explicit PRG mapping $\zo^m$ to $\zo^n$ that $\eps$-fools size-$s$ $n$-variable circuits in $\calC.$)
\end{theorem}

The existence of explicit designs is well known, in particular we recall the following:

\begin{lemma} [Problem 3.2 of \cite{Vadhan12}] \label{lem:design} There is a deterministic algorithm which, for any $r,s \in \N$, runs in time $\poly(m,s)$ and outputs an explicit $(m,r,\ell,s)$-design with $m=O(r^2/s)$ and $\ell \leq \log s.$
\end{lemma}

\pparagraph{A PRG from the \cite{Vio06} correlation bound.}  Viola \cite{Vio06} gives an explicit function $f: \zo^r \to \zo$ and shows that for every constant $d$ there is a constant $c_d$ such that $f$ is $(r^{c_d \log r}, r^{-c_d \log r})$-hard for $\SYM \circ \acz_d$.
Fix any $d$.  Given values for $\eps,s$ let us set the parameters
\[
\ell = \log s, \quad \quad r = 2^{10 \cdot \sqrt{{\frac 1 {c_d}} \log (s / \eps)}}.
\]
It is straightforward to verify that $s \cdot 2^\ell \leq r^{c_d \log r}$ and $\eps/s \geq r^{-c_d \log r}.$  By Lemma \ref{lem:design} there is an explicit $(m,r,\ell,s)$-design with $m =O( r^2 / \ell ) = 2^{O(\sqrt{{\frac 1 {c_d}} \log (s/ \eps)})}$, so applying the Nisan-Wigderson generator, we get that for $s \geq n$, there is an explicit PRG $G: \zo^m \to \zo^n$ with seed length $m = 2^{O(\sqrt{\log(s/\eps)})}$ that $\eps$-fools $n$-variable size-$s$ circuits in $\acz_d.$

\pparagraph{A PRG from the \cite{LS11} correlation bound.}  Lovett and Srinivasan \cite{LS11} give an explicit function $f: \zo^r \to \zo$ such that for every constant $d$ there is a constant $c_d$ such that $f$ is $(r^{c_d \log \log r}, \exp(-r^{1-o(1)})$-hard for $\SYM \circ \acz_d$.  We proceed as above but now choosing
\[
\ell = \log s, \quad \quad r = 2^{{\frac {10} {c_d}} \cdot {\frac {\log s}{\log \log s}}} + \left(\log (s/\eps)\right)^{1+o(1)}.
\]
It is straightforward to verify that $s \cdot 2^\ell \leq r^{c_d \log \log r}$ and $\eps/s \geq \exp(-r^{1-o(1)}).$  By Lemma \ref{lem:design} there is an explicit $(m,r,\ell,s)$-design with $m =O(r^2 / \ell) = 2^{O(\log s / \log \log s)} \cdot ( \log (1/ \eps))^{2+o(1)},$ so applying the Nisan-Wigderson generator, we get that for $s \geq n$, there is an explicit PRG $G: \zo^m \to \zo^n$ with seed length $m = 2^{O(\log s / \log \log s)} + ( \log (1/ \eps))^{2+o(1)}$ that $\eps$-fools $n$-variable size-$s$ circuits in $\SYM \circ \acz_d.$

For $\THR \circ \acz_d$, \cite{LS11} show that the same function $f$ is $(r^{c_d \log \log r}, \exp(-r^{1/2-o(1)}))$-hard for $\THR \circ \acz_d$; a similar analysis to the above gives an explicit PRG $G: \zo^m \to \zo^n$ with seed length $m = 2^{O(\log s / \log \log s)} + ( \log (1/ \eps))^{4+o(1)}$ that $\eps$-fools $n$-variable size-$s$ circuits in $\THR \circ \acz_d.$

\pparagraph{A PRG from our Theorem~\ref{thm:cor-bound}:  Proof of Theorem~\ref{thm:main}.}  Theorem~\ref{thm:cor-bound} gives an explicit
$f: \zo^r \to \zo$ and $\tau>0$ such that for all $d$, $f$ is $(r^{\tau \log r}, \exp(-r^{0.499}))$-hard for $\{\SYM,\THR\} \circ \acz_d.$  This time we choose
\[
\ell = \log s, \quad \quad r = 2^{10 \cdot \sqrt{ {\frac 2 \tau} \log s}} + \left(\log (s/\eps)\right)^{2.005}.
\]
We have $s \cdot 2^\ell \leq r^{\tau \log r}$ and $\eps/s \geq \exp(-r^{0.499}),$ so we get that for $s \geq n$, there is an explicit PRG $G: \zo^m \to \zo^n$ with seed length $m = O(r^2 / \ell) = 2^{O(\sqrt{\log s})} + ( \log (1/ \eps))^{4.01}$ that $\eps$-fools $n$-variable size-$s$ circuits in $\{\SYM,\THR\} \circ \acz_d.$

\pparagraph{A PRG from our Theorem~\ref{thm:cor-bound-many-gates}:  Proof of Corollary~\ref{cor:prg-many-gates}.}
Finally, Theorem~\ref{thm:cor-bound-many-gates} gives an explicit
$f: \zo^r \to \zo$ and $\tau>0$ such that for all $d$, $f$ is $(r^{\tau \log r}, \exp(-r^{0.499}))$-hard for the class of depth-$d$ circuits over $\zo^r$ that contain $r^{0.249}$ many $\SYM$ or $\THR$ gates.  We choose $\ell,r$ as above, so similar to the above, we get that there is an explicit PRG $G: \zo^m \to \zo^n$ with seed length $m = O(r^2 / \ell) = 2^{O(\sqrt{\log s})} + ( \log (1/ \eps))^{4.01}$ that $\eps$-fools $n$-variable size-$s$ depth-$d$ circuits with at most $2^{c \sqrt{\log s}}$ many $\SYM$ or $\THR$ gates.

\section{Proof of Theorem \ref{thm:cor-bound-many-gates}:  Handling multiple $\SYM$ and $\THR$ gates} \label{sec:many-gates}

We prove Theorem \ref{thm:cor-bound-many-gates} via a slight variant of Theorem~\ref{thm:cor-bound} and an argument from \cite{LS11} (a related argument appears in a somewhat different form in \cite{Vio06}).  The variant of Theorem~\ref{thm:cor-bound}, stated as Theorem~\ref{thm:cor-bound-analogue} below, is proved by combining ingredients (1), (2) and (3) as in Section \ref{sec:struc}, but now with the aim of proving a correlation bound against $\ANY_u \circ \{\SYM,\THR\} \circ \acz_d$ circuits rather than $\{\SYM,\THR\} \circ \acz_d$ circuits (where here and throughout this appendix we take $u := n^{0.249}$).  As we describe at the end of this section, once this correlation bound against $\ANY_u \circ \{\SYM,\THR\} \circ \acz_d$ is in place, the extension to circuits with $n^{0.249}$ many $\SYM$ or $\THR$ gates directly follows using an argument from \cite{LS11}.

In more detail we have:  (throughout the following the values of $m,k,r$ are as they were before)

\begin{lemma} [Lemma \ref{lem:fair-1} analogue] \label{lem:fair-1-analogue}
Fix $u := n^{0.249}$ and let $F$ be an $\ANY_{u} \circ \{\SYM,\THR\} \circ \acz_d$ circuit where each of the $u$ $\{\SYM,\THR\} \circ \acz_d$ subcircuits of $F$ has size at most $s=n^{\tau \log n}$.
There is a fair distribution $\calR$ over restrictions $\brho \in \{0,1,\ast\}^n$ such that the following holds:  With probability
$1-\gamma_{\SL} = 1-\exp(-\Omega_d(\sqrt{n/\log n}))$ over the draw of $\brho \leftarrow \calR$,
it is the case that  $F \uhr \brho$ belongs to the class $\mathscr{F}_{\text{simple},\,u} := \ANY_u \circ  \{\SYM, \THR\} \circ \AND_{k}$.
\end{lemma}

The proof is almost identical to that of Lemma~\ref{lem:fair-1}, with $\ANY_{u} \cdot \{\SYM,\THR\}$ taking the place of $\{\SYM,\THR\}$ throughout the argument.  Now in Corollary~\ref{cor:bash} the gate $\G$ corresponds to $\ANY_u \circ \{\SYM,\THR\}$ (rather than to just $\{\SYM,\THR\}$ as earlier) and the total circuit size of $F$ is $us$ rather than $s$ (leading to $us \cdot 2^{-t}$ rather than $s \cdot 2^{-t}$ on the RHS of the Corollary~\ref{cor:bash} bound), but this is swallowed up by the slack in the inequalities leading to (\ref{eq:gamma-SL}).

\begin{lemma} [Lemma \ref{lem:retain-structure} analogue] \label{lem:retain-structure-analogue}
Fix $u := n^{0.249}$ and let $F$ be an $\ANY_{u} \circ \{\SYM,\THR\} \circ \acz_d$ circuit where each $\{\SYM,\THR\} \circ \acz_d$ subcircuit has size at most $s=n^{\tau \log n}$. The fair distribution $\calR$ over restrictions $\brho \in \{0,1,\ast\}^n$ from Lemma \ref{lem:fair-1-analogue} satisfies the following:  With probability $1-\gamma_{\SL} - \gamma_{\target}$ over a draw of $\brho \leftarrow \calR$, both of the following hold:

\begin{enumerate}

\item [(i)] $F \uhr \brho$ belongs to $\mathscr{F}_{\text{simple},\,u} = \ANY_u \circ  \{\SYM, \THR\} \circ \AND_{k}$; and

\item [(ii)] $\RW_{m,k,r} \uhr \brho$ contains a perfect copy of $\GIP_{m/2,k+1}$.

\end{enumerate}
\end{lemma}

The proof of Lemma~\ref{lem:retain-structure-analogue} is unchanged from Section \ref{sec:target}.

\begin{theorem} [Theorem \ref{thm:cor-bound} analogue] \label{thm:cor-bound-analogue}
Fix $u := n^{0.249}.$  There is an absolute constant $\tau>0$ and an explicit $\poly(n)$-time computable function $H : \zo^n \to \zo$ with the following property:  for any constant $d$, for $n$ sufficiently large, for $F$ an $\ANY_{u} \circ \{\SYM,\THR\} \circ \acz_d$ circuit  where each $\{\SYM,\THR\} \circ \acz_d$ subcircuit has size at most $s=n^{\tau \log n}$, we have 
\[
\Pr_{\bx \leftarrow \zo^n}[F(\bx) = \RW_{m,k,r}(\bx)] \leq {\frac 1 2} + \exp(-\Omega(n^{0.249})).
\]  
\end{theorem}

The proof, using Lemmas \ref{lem:fair-1-analogue} and \ref{lem:retain-structure-analogue},
  is virtually identical to the proof of Theorem~\ref{thm:cor-bound} using Lemmas \ref{lem:fair-1} and \ref{lem:retain-structure}.  The only difference is that we use the obvious extensions of Fact~\ref{fact:HG} and Theorem~\ref{thm:nisan-PTF} to $\ANY_u \cdot \SYM \circ \AND_k$ circuits and $\ANY_u \cdot \THR \circ \AND_k$ circuits respectively; these extensions are stated for completeness below.

\begin{fact}[Fact \ref{fact:HG} analogue]
\label{fact:HG-analogue}
Let $f : \zo^n\to\zo$ be a Boolean function computed by a size-$s$ $\ANY_u \circ \SYM\circ\AND_k$ circuit. Then for any partition of the $n$ inputs of $f$ into $k+1$ blocks, there is a deterministic NOF $(k+1)$-party communication protocol that computes $f$ using $u \cdot O(k\log s)$ bits of communication.
\end{fact}

\begin{theorem}[Theorem \ref{thm:nisan-PTF} analogue] \label{thm:nisan-PTF-analogue}
Let $f : \zo^n\to\zo$ be a Boolean function computed by a $ANY_u \circ \THR\circ\AND_k$ circuit. Then for any partition of the $n$ inputs of $f$ into $k+1$ blocks, there is a randomized NOF $(k+1)$-party communication protocol that computes $f$ with error $\gamma_{\err}$ using $u \cdot O(k^3\log n\log(n/\gamma_{\err}))$ bits of communication. 
\end{theorem} 

Finally, the correlation bound Theorem~\ref{thm:cor-bound-many-gates} follows from Theorem \ref{thm:cor-bound-analogue} exactly as Theorem~6 of \cite{LS11} follows from Lemma~3 of that paper.

\end{document}